\documentclass[acmsmall]{acmart}
\renewcommand\footnotetextcopyrightpermission[1]{} 
\usepackage[english]{babel}

\newif\ifcomments
\commentsfalse

\usepackage{amsmath}

\usepackage{amsthm}
\usepackage{graphicx}
\usepackage{comment}
\usepackage{yfonts}

\usepackage{algorithm}
\usepackage[noend]{algpseudocode}
\usepackage{xcolor}
\usepackage{tabularx}

\newtheorem{theorem}{Theorem}[section]
\newtheorem{lemma}[theorem]{Lemma}
\newtheorem{definition}[theorem]{Definition}

\newtheorem{conjecture}{Conjecture}

\def\Q{\mathcal{Q}}

\newcommand{\refthm}[1]{Theorem~\ref{thm:#1}}
\newcommand{\refsec}[1]{Section~\ref{sec:#1}}

\newcommand{\BO}[1]{{ O}\left(#1\right)}

\newcommand{\BTO}[1]{\tilde{ O}\left(#1\right)}

\newcommand{\BT}[1]{{\Theta}\left(#1\right)}
\newcommand{\BOM}[1]{\Omega\left(#1\right)}
\newcommand{\BTM}[1]{\tilde{\Omega}\left(#1\right)}

\newcommand{\E}{\mathbb{E}}

\usepackage{algorithm}

\title{How many users have been here for a long time? Efficient solutions for counting long aggregated visits}

\author{Peyman Afshani}
\affiliation{%
  \institution{Aarhus University}
  \country{Denmark}
}

\author{Rezaul Chowdhury}
\affiliation{%
  \institution{Stony Brook University}
  \country{USA}
}

\author{Inge Li Gørtz}
\affiliation{%
  \institution{Technical University of Denmark}
  \country{Denmark}
}

\author{Mayank Goswami}
\affiliation{%
  \institution{City University of New York, Queens College}
  \country{USA}
}

\author{Francesco Silvestri}
\authornote{Corresponding author \url{francesco.silvestri@unipd.it}.}
\affiliation{%
  \institution{University of Padova}
  \country{Italy}
}

\author{Mariafiore Tognon}
\affiliation{%
  \institution{University of Padova}
  \country{Italy}
}

\begin{abstract}
This paper addresses the \emph{Counting Long Aggregated Visits} problem, which is defined as follows. We are given $n$ users and $m$ regions, where each user spends some time visiting some regions. For a parameter $k$ and a query consisting of a subset of $r$ regions, the task is to count the number of distinct users whose aggregate time spent visiting the query regions is at least $k$. This problem is motivated by queries arising in the analysis of large-scale mobility datasets.

We present several exact and approximate data structures for supporting counting long aggregated visits, as well as conditional and unconditional lower bounds. First, we describe an exact data structure that exhibits a space-time tradeoff, as well as efficient approximate solutions based on sampling and sketching techniques. We then study the problem in geometric settings where regions are points in $\mathbb{R}^d$ and queries are hyperrectangles, and derive exact data structures that achieve improved performance in these structured spaces.
\end{abstract}

\ccsdesc[500]{Theory of computation~Data structures and algorithms for data management}
\ccsdesc[300]{Theory of computation~Sketching and sampling}
\ccsdesc[300]{Theory of computation~Lower bounds and information complexity}

\keywords{Exact and approximate counting, sketching, geometric algorithms, lower bounds}
\begin{document}
\maketitle
\thispagestyle{plain}
\pagestyle{plain}
\makeatletter
\def\ps@plain{%
  \let\@oddhead\@empty
  \let\@evenhead\@empty
  \def\@oddfoot{\hfil\thepage\hfil}
  \def\@evenfoot{\hfil\thepage\hfil}
}
\makeatother

\section{Introduction}
A common data type in large-scale mobility datasets consists of geo-located data with time information \cite{Pelekis14}.
For the sake of simplicity, we can view these datasets as consisting of triplets, where each triplet $(u,r,t)$ denotes that a given user $u$ has spent some time $t$ in a given region $r$.
According to the different technologies used to collect data, the datasets might differ in how a region is denoted.
One common approach leverages mobile network logs (i.e., CDR), which record the cell towers used by each smartphone at a given time. 
These datasets provide broad population coverage; however, they cannot pinpoint a user’s exact location and instead only indicate that the user was somewhere within the coverage area of a given cell tower. 
One alternative approach relies on applications that periodically record GPS coordinates together with timestamps, such as navigation apps. 
Although these data might not provide the same statistical guarantees as network logs, they provide precise GPS coordinates.
Mobility datasets are typically very large: for instance, the telecommunications company Vodafone collects approximately 30 billion records per day just from users in Italy \cite{Vodafone23}, while TomTom collects 61 billion GPS data points per day from their navigation devices \cite{TomTom_MapUpdate_2019}.

These mobility datasets are used to extract information on human behaviors that are of interest, for instance, for urban development, economics, and tourism.
Common queries on these datasets are, for example: "\textit{How many commuters have spent more than 1 hour stuck in the highway?}", "\textit{How many unique individuals have visited the city center for more than an hour?}", "\textit{How many unique tourists have visited Venice and Florence for more than 2 days?}" (e.g. \cite{Louail2014CongestedTravel,Montoliu2013,Cavallo2022Exploring}).
    As users move, the queries require recognizing the same user in different regions and computing the overall time in the desired area.
Consider the input triplets $(u_1,r_1, 20 h)$, $(u_2,r_1, 15 h)$, $(u_1,r_2, 15 h)$, $(u_3,r_2,30 h)$, $(u_2,r_3,20 h)$ and the area of interest defined by the regions $\{r_1, r_2\}$:
then, the two users $u_1, u_3$ have spent at least 30h in the area, while $u_2$ does not reach the required minimum time.
The area specified by these questions can be translated into a list of cell towers that cover the area of interest for network logs, or as bounding boxes for GPS coordinates. 
Moreover, for privacy reasons, these queries do not ask to list users, but rather to provide counts.
Note that similar queries arise in different scenarios, such as when evaluating user engagement on websites (e.g. \cite{8609680}).

In this paper, we provide formal definitions of the above problem, which we name \emph{Counting Long Aggregated Visits} or just  \emph{CLAV} for short.
We first provide a generic definition that does not make any assumptions about region properties and can also be used for non-geometric objects (e.g., visits to a webpage).
Then we define \emph{Geometric-CLAV}, a variant in which users and regions belong to the Euclidean space \emph{$\mathbb{R}^d$}.
We provide efficient approaches to compute exact and approximate solutions to these problems, as well as unconditional and conditional lower bounds. 
We now describe the CLAV and Geometric-CLAV problems and our results, which are also summarized in Table \ref{tab:results}.

\begin{table}
\footnotesize
\begin{tabularx}{\linewidth}{|m{11em}|m{10em}|m{5.5em}|X|m{4em}|}
\hline
 \textbf{Type of CLAV} & \textbf{Space} &\textbf{Query time} & \textbf{Guarantees} & \textbf{Theorem}\\
\hline\hline
$(k,r)$-CLAV, exact solution  
    & $S$ 
    & $\BTM{N/S^{1/r}}$ 
    & Conditional lower bound 
    & \ref{thm:lowergeneral} \\
\hline
$(k,r)$-CLAV, exact solution
    & $\BO{ (N/\lambda)^{r}n+N}$ 
    & $\BTO{r \lambda}$ 
    & For any $\lambda>0$ 
    & \ref{thm:exactalg}\\
\hline
$(k,r)$-CLAV, approximate solution, sampling based  
    & $\BO{N}$ 
    & $\BO{\frac{r^3}{\varepsilon^2}\log n}$
    &  ${\varepsilon}n_Q$ additive error with high probability, $\varepsilon >0$  
    &  \ref{thm: nqk_estimate}\\
\hline
$(k,r)$-CLAV, approximate solution, sketch based
    & $\BO{m \varepsilon^{-1} \log^2 n \log r}$ bits
    & $\BO{\frac{1}{\varepsilon} \log^2 n }$ 
    & $\BT{n_{Q,k}+\varepsilon n_Q}$ additive error with high probability, $\varepsilon = \BOM{{n_{Q,k}}/{n_Q}}$; it might count users with aggregate time $\geq k(1-1/r)$  
    & \ref{thm:sketch}\\
\hline
Geometric CLAV, exact 
    & $\BOM{\min\{n,\left(\frac{m}{2d}\right)^{2d}\}}$ bits  
    & - & Unconditional lower bound for any query time
    & \ref{thm:geolbuncon}\\
\hline
Geometric CLAV, exact
    &  $ \Omega( m^{2d-1-o(1)})$
    &  $\BTO{1}$
    & Conditional lower bound 
    & \ref{thm:condlb1}\\
    \hline
Geometric CLAV, exact, $d=2$
    &  $\Omega(m^{3-\varepsilon})$ preprocessing time
    &  $O(m^{1-\varepsilon})$
    & Conditional lower bound for any $\varepsilon>0$
    & \ref{thm:condlb2}\\
    \hline
Geometric CLAV, exact
    &  $\BO{\min\{Nm^{2d-2},m^{2d}\}}$ 
    &  $O(\log_{w} n_{Q,k})$
    & For points in $\mathbb{R}^d$
    & \ref{thm:geometric}\\
    \hline
\end{tabularx}
\caption{\label{tab:results} Our main results are summarized here. Bounds are provided in a simplified version for readability. Space is in words, unless differently stated. Recall that $n$ is the number of users, $m$ is the number of regions, $N=|T|$, $n_Q$ is the number of distinct users in the query region $Q$ and $n_{Q,k}$ is the number of distinct users in the query region $Q$ with aggregated time $\geq k$ (the required output), and $d=\BO{1}$  is the dimension of the points in the geometric version. $w$ is the word-size, which is at least $\log(n+m)$.
The $\BTO\cdot$ and $\tilde \Omega(\cdot)$ notations hide polylogarithmic factors.}
\end{table}

\subsection{$(k,r)$-CLAV}\label{sec:krclav}
The first definition of the counting long aggregated visits problem for general objects is as follows:
\begin{definition}\label{def:general}
Let $U=\{u_0, \ldots u_{n-1}\}$ be a set of $n$ distinct users, $R=\{r_0, \ldots r_{m-1}\}$ be a set of $m$ regions, and let $T=\{(i,j,\tau_{i,j}), \text{with } 0\leq i <n, 0\leq j<m,  \tau_{i,j}>0\}$ be a set of $N$ triplets where each triplet $(i,j,\tau_{i,j})$ denotes that user $u_i$ has spent time $\tau_{i,j}$ in region $r_j$; we assume $\tau_{i,j}=0$ if there are no triplets $(i,j,\cdot)\in T$.
Let $k>0$ and  $1\leq r \leq m$ be two given parameters.
The \emph{$(k,r)$-CLAV} requires constructing a data structure, given the $U, R, T$ as input, that returns for any query $Q\subseteq R$ with $|Q|=r$ the number $n_{Q,k}$ of users who have spent no less than a total time $k$ in regions denoted by $Q$. 
That is:
\[
n_{Q,k} = \left\vert \left\{u_i\in U: \sum_{j\in Q} \tau_{i,j} \geq k\right\} \right\vert.
\]
\end{definition} 

Without loss of generality, we assume that there are no two triplets $(i,j,k_1), (i,j,k_2)$ in $T$ with the same $i$ and $j$ values; this property easily follows by replacing the triplets with the $(i,j,k_1+k_2)$.
We use the following notation: $N$ is the number of triplets in $T$ (i.e., $N=\|T\|$); $R_j$ denotes the set of users $u_i$ with $\tau_{i,j}>0$ (i.e. $R_j=\{i: \tau_{i,j}>0)\}$) with $j\in [r]$;\footnote{We define $[n]=\{0, 1,\ldots n-1\}$.} $N_Q$ is the number of triplets in the query regions $Q$ (i.e., $N_Q= \sum_{j\in Q} |R_j|$); $n_Q$ is the number of distinct users in the query regions $Q$. Finally, space requirements are in memory words unless differently specified; we assume each memory word to contain $w=\BOM{\log (n+m)}$ bits.

It is easy to see that a simple exact solution to a query $Q$ requires $\BO{N_Q}$ query time by scanning the users in the $r$ sets in $Q$.
On the other hand, we can achieve a $\BO{1}$ query time if all 
the possible queries are precomputed with space requirements exponential in $r$.
In Section \ref{sec:exactupper} and Section \ref{sec:lb}, we complete the tradeoff and derive a tight lower bound by leveraging results on set intersection (see e.g. \cite{goldstein2017conditional}).
Unlike set intersection, the $(k,r)$-CLAV does not require all users to belong to all regions in a query; for instance, given a query $Q$, it suffices for a user $u_i$ to appear with time $k/2$ in two regions $r_1, r_2\in Q$ to be counted for any $r\geq 2$.

In Section \ref{sec:approxsampling}, we analyze approximate solutions. 
We first provide a linear space solution with $\BO{{r^3}{\varepsilon^{-2}}\log (n)}$ query time for any $\varepsilon>0$: it provides an unbiased estimator of $n_{Q,k}$ and, with high probability,
it achieves a maximum additive error of $\BO{\varepsilon n_Q}$.
We then provide a compressed data structure that uses a sketch based on the Flajolet-Martin and Count-Min sketches to summarize each region.
The data structure needs $\BO{m \varepsilon^{-1} \log^2 n \log r}$ bits and  $\BO{r\varepsilon^{-1}  \log^2 n }$  query time.
The improvement in space and query time comes with a slightly lower quality of the results: the estimator might include users with aggregate time $\geq k(1-1/r)$, and the maximum error is $\BT{n_{Q,k}+\varepsilon n_Q}$ with high probability for $\varepsilon =\BOM{{n_{Q,k}}/{n_Q}}$; moreover, if a rough estimate of $n_{Q,k}$ is known, we obtain a 3-approximation with $\varepsilon =\BT{{n_{Q,k}}/{n_Q}}$.

We observe that a simple sketch can also be derived from the results of detecting heavy hitters in streams, as users that appear $k$ times in the $Q$ regions are users that appear with a frequency of $\phi \tilde T_Q$, with $\phi= k/{\tilde{T}_Q}$ and $\tilde T_Q = \sum_{(i,j,\tau_{i,j})\in T, j\in Q} \tau_{i,j}$.
However, this would provide a sketch where the space requirements depend on the aggregated times of all users in the query region; on the contrary, our solution is independent of $\tilde T_Q$.

\subsection{Geometric-CLAV}
We now present the definition for the geometric case. 
We represent each region by a point in $\mathbb{R}^d$, for instance, a Point-Of-Interest (POI) in a city; the query is represented by a hyperrectangle.

\begin{definition}[Geometric-CLAV]\label{def:geometric}
Let $U=\{u_0, \ldots u_{n-1}\}$ be a set of $n$ distinct users, let
$R=\{x \in \mathbb{R}^d\}$ be a set of $m$ points in $\mathbb{R}^d$, 
and let $T=\{(i, x, \tau_{i,x}) \text{with } \tau_{i,x} > 0,  0\leq i <n, x\in R\}$ 
be a set of $N$ triplets where each triplet $(i,x,\tau_{i,x})$ 
denotes that user $u_i$ has spent time $\tau_{i,x}$ in point $x$; we assume $\tau_{i,x}=0$ if there are no triplets $(i,x,\cdot)\in T$.
Given $U,R,T$, the \emph{Geometric-CLAV} requires constructing a data structure that, given a query $Q$ denoted by a hyperrectangle in $\mathbb{R}^d$, returns the number $n_Q$ of users that have spent no less than $k$ in query $Q$, that is:
\[
n_{Q,k} = \left\vert \left\{u_i\in U: \sum_{x\in Q, \tau_{i,x} \in T }    \tau_{i,x} \geq k\right\} \right\vert.
\]
\end{definition}

We now describe the results for this version.
We use similar notation as for the $(k,r)$-CLAV problem
and, for notational simplicity, we assume $d$ to be a constant.

We first provide, in Section \ref{sec:geolbuncon}, an unconditional space lower bound of $\min\{n,(m/2d)^{2d}\}$ bits, regardless of the query time.
When $n > (m/2d)^{2d} $, the lower bound is matched by a solution with a look-up table using tabulation (Theorem~\ref{thm:tabulation}) up to polylog terms.

Then, in Sections \ref{sec:condlb1} and \ref{sec:condlb2}, we prove two stronger conditional lower bounds. 
The first bound is an $\Omega(m^{2d-1-o(1)})$ space lower bound on any data structure with polylogarithmic query time, assuming a fast-query version of the $r$-SetDisjointness Conjecture in \cite{goldstein2017conditional}.
For a planar case, we prove an alternative conditional lower bound: if $n=m$ and $N=O(m^2)$, any data structure must either have $\Omega(m^{1-\varepsilon})$ query time or $\Omega(m^{3-\varepsilon})$ preprocessing, for any $\varepsilon>0$, assuming the Boolean-Matrix-Multiplication conjecture~\cite{subcubic,bansal2009regularity}.

Finally, in Section \ref{sec:ubgeometric}, we provide a data structure that uses a reduction to the colored dominance counting problem \cite{gao2023adaptive}.
The data structure requires $\BO{\min\{Nm^{2d-2},m^{2d}\}}$ memory words and $O(\log_{w} n_{Q,k})$ query time.
This data structure almost matches the conditional lower bound obtained in \refsec{condlb1} (up to $n^{o(1)}$ factors).

\subsection{Previous works}
Research related to our problem spans several different domains: indeed, similar tasks have been tackled as \textit{Spatiotemporal Queries}, \textit{Heavy Hitters}, or \textit{Distinct Counting}, both from an exact and approximate standpoint.
In what follows, we aim to outline the differences between our problem and existing formulations in the literature.

A large body of work is devoted to solving spatiotemporal aggregation queries: the typically posed problems consist of \emph{timestamp} or \emph{interval} queries, which ask to report, count, or compute aggregates over objects that appear within a specific (rectangular) spatial region during a specified time window. This is essentially equivalent to focusing on all users that appear in an area for more than $k$ time with $k=0$.
Examples of reporting tasks include \cite{tao2001mv3r}, \cite{romero2012smo},\cite{gutierrez2005spatio}.  Counting variants are also common: for instance, \cite{tao2004spatio} tackles the problem of counting users that satisfy a specific spatiotemporal query, but other variations involve computing the average number of visitors per timestamp \cite{tao2002aggregate} or reporting the top-$k$ temporal intervals per cardinality \cite{xu2023topkaggregate}. Additional works focus instead on the computation of feature-based aggregates of users that appear in a low-dimensional rectangular query, which identifies both a time interval and a geographic area \cite{timko2009sequenced}.
Note that this kind of query is fundamentally different from ours, as no non-trivial requirement on the users' dwell-time is imposed.
As a consequence, since the typical input data for interval queries consists of snapshots capturing the presence of a user in a specific location at a specific time, the employed data structures do not easily adapt to our case for general $k$.

Our task involves counting distinct users who appear in a union of regions $Q$, which becomes known only at query time. This suggests the use of sketch-based approximations to speed up the task and reduce memory requirements.
When it comes to counting, several sketches for cardinality estimation have been thoroughly studied in the last decades (\cite{flajolet1985probabilistic,whang1990linear,baryossef2002countdistinct,heule2013hyperloglog,martingale,ertl2023ultraloglog}), achieving optimal space-accuracy tradeoffs (\cite{woodruff2010optimal}) and enjoying mergeability properties (i.e. allowing the sketch of a union to be obtained by combining sketches of individual sets). 
However, to the best of our knowledge, none of these works allows filtering out users that do not match a \textit{cross-regional} predicate while simultaneously reporting a cardinality estimate of the target elements. 

Similarly, even existing sketches for heavy-hitter reporting do not seem to fit. Popular Algorithms for Heavy Hitters, such as Misra-Gries (\cite{misra1982finding,demaine2002frequency,zhang2013distributed}), Sticky Sampling (\cite{manku2002sticky}) or the optimal reporting algorithm in \cite{yi2009optimaldistributed}, cannot be trivially extended to our case, as their naive readaptation leads to sketches whose space depends on the specific values of stay time in $T$, as noted at the end of Section \ref{sec:krclav}. 
As far as CountSketch and its variants (\cite{alon1996space,charikar2002finding,cormode2005improved}) are concerned, they allow recovering the (approximate) frequency of the heaviest items only by explicitly querying their keys.

Finally, existing works on monitoring thresholded data focus either on \textit{identifying or monitoring} all keys that exceed a certain threshold without simultaneously allowing for multiple different queries (\cite{li2020distributedthreshold,keralapura2006communication}); we instead aim to design simpler data structures that support counting and multiple user-defined queries, which can therefore benefit both query runtime and space complexity with respect to the blind application of existing solutions. 

As for the Geometric-CLAV, when $k=1$ and when $\tau_{i,x}$ is either 0 or 1, the problem is known as ``generalized intersection searching'' (or colored range searching)~\cite{janardan1993generalized,gupta1995further} and
it has received considerable attention since~\cite{gupta2018survey}. In the variant most relevant to us, the goal is to build a data structure for an input point set where each point is assigned a color, and given a query rectangle, the goal is to 
count the number of distinct colors in the query. The problem can be efficiently solved in 1D~\cite{janardan1993generalized,gao2021twod,gao2023adaptive}, but it is known to be difficult in dimensions 2 and above~\cite{kaplan2008efficient}.
Interestingly, our conditional lower bounds are higher than those obtained in~\cite{kaplan2008efficient}, which shows that our problem is strictly more difficult.

\section{Exact solutions to the $(k,r)$-CLAV problem}

Here, we derive almost matching upper and conditional lower bounds on the general (i.e., not geometric) $(k,r)$-CLAV problem (Definition~\ref{def:general}). We start with the lower bound.

\subsection{Lower bounds for the $(k,r)$-CLAV problem}
\label{sec:lb}

Let $S_0,\cdots,S_{m-1}$ be subsets of a universe $U$ of total size $N=\sum_{i=1}^{m} |S_{i}|$. The goal of the $r$-Set Disjointness problem is to preprocess this collection of subsets in order to answer queries of the form: given $Q \subseteq [m]$ such that $|Q|=r$, decide if $\cap_{q \in Q} S_{q}$ is empty, i.e., whether the subsets $S_{q}$ for $q \in Q$ are disjoint. For $r=2$, this is simply pairwise set intersection, which is a well-known problem with its own conjectures, the SetDisjointness and the Strong SetDisjointness conjectures (Conjectures 1 and 2 in \cite{goldstein2019hardness}). 
Our source of hardness is the Strong $r$-Set Disjointness Conjecture (Conjecture 9 in \cite{goldstein2017conditional}, who also provide a matching upper bound). 

\begin{conjecture}[\cite{goldstein2017conditional}] Any data structure for the $r$-SetDisjointness Problem that uses $S$ space and has query time $T$ must satisfy $ST^{r}=\widetilde{\Omega}(N^{r})$.
\end{conjecture}

We would like to point out that, although not stated in \cite{goldstein2017conditional}, the above conjecture can only hold for large enough universe size as there is a trivial solution (using $N$ space, just storing the input) running in time $T=r|U|$, and if $|U| \ll N$ the conjecture is clearly broken. In Section~\ref{sec:geometric} we present a version of the above conjecture when the universe size is at least polylogarithmic in $m$, and the query time is at most polylogarithmic in $m$, that we call the Polylog-UniQuery-$\beta$-SetDisjointness Conjecture (Conjecture 3 in Section \ref{sec:condlb1}).

\begin{theorem}\label{thm:lowergeneral}
    Assume Conjecture 1 holds. Then any data structure for the $(k,r)$-CLAV problem that uses $S$ space and has query time $T$ must satisfy $ST^{r}=\widetilde{\Omega}(N^{r})$.
\end{theorem}

\begin{proof}
    We will show how to use a data structure for the $(k,r)$-CLAV problem to solve the $r$-SetDisjointness Problem with the same space and time usage. Consider the input to the latter problem, i.e. the sets $S_0,\cdots, S_{m-1} \subseteq U$. By relabeling, we can assume that $U=[n]$ where $n:=|\cup_{i=1}^{m} S_{i}|$. We create $m$ regions $\{r_0,\cdots,r_{m-1}\}$ and $n$ users. For every $j \in [n]$ and every $i \in [m]$ such that $j \in S_{i}$, we create a triple $(j,i,1)$. The number of such triples is exactly $N=\sum_{i=1}^{m} |S_{i}|$. We preprocess a data structure for this instance of the $(k,r)$-CLAV problem, using $S$ space. 

    Given a query $Q \subseteq [m]$ to the $r$-SetDisjointness problem, we issue the same query $Q$ with $k=r$ to the data structure for the $(k,r)$-CLAV problem. Since any user spends at most one unit of time in any region, asking for users who spent an aggregate of $r$ time in $r$ regions is precisely asking for the size of the intersection of the corresponding sets. Therefore, if the returned count to the query for the $(k,r)$-CLAV problem is non-zero, we report that the query sets for the $r$-SetDisjointness problem are not disjoint, and report that they are disjoint otherwise.
\end{proof}

Note that our reduction actually works for the $r$-SetIntersectionCount problem, where the goal is not just to return whether the $r$ sets are disjoint or not, but the size of their intersection. We have not seen this problem in previous works. Is $r$-SetIntersectionCount harder than $r$-SetDisjointness? For $r=2$ it was shown that $r$-SetIntersectionCount has the same tight tradeoff as the $r$-SetDisjointness problem (see Appendix C.4 in \cite{goldstein2019hardness}).

However, let us examine their upper bound and see if it generalizes to $r>2$. For a given query time target $T$, \cite{goldstein2019hardness} calls a set large if its size is at least $T$, and otherwise the set is called small. This implies that there are at most $N/T$ large sets. For every pair of such sets, they precompute the answer (the count of the intersection), taking space $S= O((N/T)^2)$. They also store, for every $S_i$, $i \in [m]$, a membership hash table for the elements in $S_i$, taking $O(N)$ space in total. Given a query $|S_i \cap S_j|$, if both $S_i$ and $S_j$ are large, they can answer it in constant time. Otherwise, assume that $S_i$ is small, i.e., has at most $T$ elements. For every element in $S_i$, query in $O(1)$ time whether it is in $S_j$, and return $|S_i \cap S_j|$ for a total of $O(T)$ query time. The space-query tradeoff satisfies $ST^2=O(N^2)$. 

When $r>2$, the same strategy almost works, but now if at least one of the $r$ sets is small, we can take the smallest set in the query, having at most $T$ elements, and now we have to query all the remaining $r-1$ hash tables (for $S_q$, $q \neq i$) for each of these elements. Thus, the query time becomes $O(rT)$. For constant $r$, this gives the same tradeoff as $r$-SetDisjointness, but for $r>2$ this introduces an extra factor $r$ in the query time. Thus, we make the following conjecture.

\begin{conjecture}
Any data structure for the $r$-SetIntersectionCount problem that uses $S$ space and has query time $T$ must satisfy $ST^{r}=\widetilde{\Omega}((rN)^{r})$.
\end{conjecture}

The proof of Theorem~\ref{thm:lowergeneral} shows that Conjecture 2 applies to the $(k,r)$-CLAV problem. We now give an upper bound that matches the tradeoff of Conjecture 2.

\subsection{Upper bound for the $(k,r)$-CLAV problem}
\label{sec:exactupper}

Given a fixed $r$, we show below how to construct a data structure $\mathcal{D}_r$ to solve the exact $(k,r)$-CLAV problem for any $k$. 
The data structure builds on the division into \emph{large} and \emph{small} sets as proposed in the exact data structure for the $r$-SetIntersection problem \cite{goldstein2017conditional}. 
However, our data structure requires additional components to deal with users who only appear in a subset of the query.
We recall that the set $T$ is a set of $N$ triplets where each triplet $(i,j,\tau_{i,j})$ denotes that user $u_i$ spent time $\tau_{i,j}$ in region $r_j$, and $R_j\subseteq [n]$ denotes the set of users who spent some time in region $r_j$.

Let $\lambda > 0$ be a suitable parameter. We call a region $r_j$ \emph{large} if $|R_j|>\lambda$, and \emph{small} otherwise. Note that there cannot be more than $\min\{m, N/\lambda \}$ large regions. Let $L_{\lambda}$ and $S_{\lambda}$ be the subsets of all large and small regions in $R$, respectively.  

At construction time, we first sort $T$ with key $(i,j)$ for easily verifying (through binary search) if a given user $u_i$ spent time in a given region $r_j$ and if so, retrieve the amount of time $\tau_{i,j}$ they spent there. Then, for every subset $Q \subseteq L_{\lambda}$, we do the following. We compute the temporary set $V_{Q}$ of all tuples $( i, \tau_{i, Q} )$ with $i \in \bigcup_{r_j \in Q}{R_j}$, and set $\tau_{i, Q} = \sum_{r_j \in Q}{\tau_{i, j}}$. We create two arrays $T_Q$ and $U_Q$ from $V_{Q}$, where $T_Q$ contains all tuples of $V_{Q}$ sorted in non-decreasing order of time $\tau_{i, Q}$, and $U_Q$ has them sorted in non-decreasing order of user index $i$. Our data structure $\mathcal{D}_r$ consists of the sorted version of $T$, $R_j$ for all $r_j \in S_\lambda$, and the $T_Q$ and $U_Q$ arrays for all $Q \subseteq L_{\lambda}$.

Now, given a set $Q = Q_{l} \cup Q_s$ of query regions with $|Q| \leq r$, $Q_l \subseteq L_{\lambda}$ and $Q_s \subseteq S_{\lambda}$, we use $\mathcal{D}_r$ to find the number $n_{Q,k}$ of distinct users who spent at least $k$ time in $Q$ as described in Algorithm \ref{alg:tradeoff}. 
We first look at the precomputed solutions for $Q_l$.
Then, we look at users in the $Q_s$ regions: for each user $u$, we verify if its aggregate time is at least $k$ and we count it only if it was not already counted in the solution for $Q_l$ (i.e., the aggregate time in $Q_l$ for $u$ is already at least $k$).

\begin{algorithm}[t]
\footnotesize
\caption{Exact algorithm for $n_{Q,k}$}\label{alg:tradeoff}
\begin{algorithmic}[1]
\State Perform a binary search on $T_{Q_{l}}$ with $k$ to find the number $n_{Q_l}$ of unique users who spent at least $k$ time in $Q_l$.
\State Compute the set $R_{Q_s}$ of all tuples $( i, \tau_{i, Q_s} )$ with $i \in \bigcup_{r_j \in Q_s}{R_j}$
\State Set $\tau_{i, Q_s} = \sum_{r_j \in Q_s}{\tau_{i, j}}$ and $n_{Q_s} = 0$. 
\For{every $( i, \tau_{i, Q_s} ) \in R_{Q_s}$}
    \If{a tuple $( i, \tau_{i, Q_{l}} )$ with $\tau_{i, Q_{l}} \geq k$ does \emph{not} exist in $U_{Q_l}$} 
    \Comment{Check performed with a binary search on $U_{Q_l}$}
    \State Increment $n_{Q_s}$ provided $\tau_{i, Q_{s}} + \tau_{i, Q_{l}} \geq k$  \Comment{assume $\tau_{i, Q_{l}} = 0$ if no tuple of the form $( i, * )$ exists in $U_{Q_l}$}
    \EndIf
\EndFor
    \State \Return $n_{Q_l} + n_{Q_s}$~ $( = n_{Q, k} )$
\end{algorithmic}
\end{algorithm}

The following theorem captures the trade-off between memory and query time.
We note that, in the special case where $r =\BO{1}$, $m=\BO{n}$, $N=\BO{n^2}$ and $\lambda = \sqrt{n}$, we get a data structure with $\BTO{\sqrt{n}}$ query time and $\BO{n^{3r/2+1}}$ memory. 

\begin{theorem}\label{thm:exactalg}
The above data structure requires $\BTO{(\min\{m, N/\lambda\})^{r}r n}$ construction time and $\BTO{\min\{r \lambda, N \}}$ query time, and uses $\BO{(\min\{m, N/\lambda\})^{r}n+N}$  memory words.
\end{theorem}
\begin{proof}
The number of subsets $Q \subseteq L_{\lambda}$ with $|Q| \leq r$ is $\BO{\left(\min\{m, N/\lambda\}\right)^r}$, and we require $\BTO{r n}$ time to construct the $T_Q$ and $U_Q$ arrays for each such $Q$. Thus, $\mathcal{D}_r$ can be constructed in $\BTO{(\min\{m, N/\lambda\})^{r} r n}$ time. We require $\BO{(\min\{m, N/\lambda\})^{r}n}$ memory words to store all $T_Q$ and $U_Q$ arrays, and $\BO{N}$ memory words to store the sorted version of $T$ and the $R_j$ sets for all $r_j \in S_{\lambda}$. Thus, $\mathcal{D}_r$ occupies $\BO{(\min\{m, N/\lambda\})^{r}n+N}$ memory words. The cost of a query is dominated by the time required to compute $R_{Q_s}$ and $n_{Q_s}$, both of which can be done in $\BTO{\min\{r \lambda, N \}}$ time.
\end{proof}

\section{Approximate solutions for the $(k,r)$-CLAV problem}

Here, we derive approximate solutions for the generic $(k,r)$-CLAV problem  (Definition~\ref{def:general}). 
\subsection{Approximation via Sampling}\label{sec:approxsampling} 
Let $Q$ be a set of $r$ query regions and denote with $T_Q\subset T$ the set of triples appearing in region $Q$.
Let $N_Q$ and $n_Q$ indicate the number of triples and the number of distinct users in $Q$, respectively. 
In order to compute $n_{Q,k}$ exactly, one needs to read all $N_Q$ existing triples from the query area: however, if an approximate estimate is sufficient, we can outline a sampling-based query procedure that allows us to estimate $n_{Q,k}$ quickly within reasonable error bounds with a user-specified probability.
First of all, we note that for any query set $Q$, we can rewrite $n_{Q,k}$ as 
\begin{equation}\label{eq:sum_nq}
    n_{Q,k} \,=\, \sum_{(i,j,\tau_{i,j})\in T_Q} \frac{\phi_i}{c_i}
\end{equation}
where $c_i$ is the number of regions in $Q$ containing user $i$ (i.e., the number of distinct triples associated with $i$ in $T_Q$), while $\phi_i$ is an indicator variable that takes on the value $1$ if $\sum_{r_j\in Q}\tau_{i,j}\geq k$, and 0 otherwise.
The scaling factor $c_i$ is needed since a user $i$ with $\phi_i=1$ appears in $c_i$ regions of $Q$.
In what follows, we will denote estimators for $n_{Q,k}$  by $\hat{n}_{Q,k}$.

Since this procedure aims at speeding up query time only, we admit a preprocessing that stores triples as follows. We place each triple $(i, j, \tau_{i, j}) \in T$ into an array $\overline{T}$ and sort $\overline{T}$ by $(i, j)$ which will allow us to quickly determine, through a binary search, whether a given user $u_i$ spent time in a given region $r_j$, and, if so, to retrieve $\tau_{i, j}$.

After the preprocessing described above, the query procedure for deriving an unbiased $\hat{n}_{Q,k}$ is outlined in Algorithm \ref{alg:cap}. Specifically, we define a sample size $s$, which allows us to obtain a good approximation to $n_{Q,k}$, as per Theorem \ref{thm: nqk_estimate}. We sample $s$ triplets uniformly at random (with replacement) from $T_Q$. For each sampled triplet $(i, j, *)$, we perform a binary search on $\overline{T}$ for each $r_q \in Q \setminus \{r_j\}$ to determine whether $u_i$ spent time in $r_q$ and, if so, to retrieve $\tau_{i, q}$. If $\sum_{r_q \in Q}{\tau_{i, q}} \geq k$, then we set $\phi_i = 1$; otherwise, we set $\phi_i = 0$. By summing the $\frac{\phi_{i}}{c_{i}}$ values over all sampled triplets and rescaling appropriately (line \ref{line:rescale} of Algorithm \ref{alg:cap}), we obtain the estimator $\hat{n}_{Q,k}$.

\begin{algorithm}[t]
\footnotesize
\caption{Sampling-Based Estimate of $n_{Q,k}$}\label{alg:cap}
\begin{algorithmic}[1]

\State $s\gets \frac{r^2}{2\varepsilon^2}\ln(2/\delta)$
\State $z\gets 0 $
\For {$l \gets 1 \text{\textbf{ to }} s$}
\State $(i,j, \tau_{i,j})\gets \text{random triple from $T_Q$} $
\State $t \gets \tau_{i,j}$,~ $c_i \gets 1$
\For {each $r_q \in Q - \{r_{j}\}$}
\If {$(i, q, \tau_{i, q}) \in {\overline{T}}$} \Comment{found via binary search}
\State $t \gets t + \tau_{i, q}$,~ $c_i \gets c_i + 1$
\EndIf
\EndFor
\If {$t \geq k$} $\phi_i \gets 1$
\Else~ $\phi_i \gets 0$
\EndIf
\State $z\gets z+ \frac{\phi_i}{c_{i}}$
\EndFor

\State $\hat{n}_{Q,k}\gets z \cdot \frac{N_Q}{s} $ \label{line:rescale}
\State \Return $\hat{n}_{Q,k}$
\end{algorithmic}
\end{algorithm}

The following claim shows that $\hat{n}_{Q,k}$ obtained at line \ref{line:rescale} of Algorithm \ref{alg:cap} is unbiased, and it outlines the quality that can be achieved by sampling $s$ entries within the region.
\begin{theorem}\label{thm: nqk_estimate}
Let $0 < \delta < 1$ and $0 < \varepsilon < 1$ be the confidence and accuracy parameters, respectively, and let $Q$ be a set of $r$ query regions. Then, for any $s \geq \frac{r^2}{2\varepsilon^2}\ln(2/\delta)$, one can compute an unbiased estimator $\hat{n}_{Q,k}$ of $n_{Q,k}$ by querying $s$ triplets drawn uniformly at random (with replacement) from $T_Q$ with the guarantee that $\hat{n}_{Q,k}\in [n_{Q,k}-{\varepsilon} n_{Q}, n_{Q,k}+{\varepsilon} n_{Q}]$ will hold with probability at least $1-\delta$.
\end{theorem}
\begin{proof}
Let $\hat{n}_{Q,k}$ be computed according to Algorithm \ref{alg:cap}. We need to show that it is an unbiased estimate of $n_{Q,k}$ and to determine the smallest sample size required to achieve the stated error bounds.
Let $s>0$ denote the sample size. Then we will query $s$ triplets (out of the $N_Q$ triplets appearing in $T_Q$) to obtain our bounded error estimate of $n_{Q,k}$. 
For $1 \leq q \leq s$, let the $q$-th sample drawn from $T_Q$ have the form $(i_q, *, * )$, that is, it belongs to user $u_{i_q}$. For each $q$ we then define the following random variable:
\begin{align*}
    Y_q = \left\{\begin{array}{ll}
        \frac{1}{c_{i_q}} &\text{ if } \phi_{i_q} = 1; \\
        0 &\text{ otherwise. }   
    \end{array}\right. 
\end{align*}
Let $Z = \sum_{q=1}^s Y_q$ and note that this will be the value stored in $z$ at line \ref{line:rescale} of Algorithm \ref{alg:cap}.
Observe that $\E[Z]= \E[\sum_{q=1}^s Y_q]= s \E[Y_q] $ and 
$\E[Y_q] = 0 \cdot \Pr[Y_q = 0 ]+ \sum_{h=1}^r\frac{1}{h}\Pr\left[Y_q = \frac{1}{h}\right]= \sum_{h=1}^r\frac{1}{h}\cdot\frac{v_h}{N_Q},$ 
where $v_h$ is defined as the number of triples in $T_Q$ corresponding to users who visited exactly $h$ regions of $Q$ and spent a total of at least $k$ units of time in those regions.
It follows that $\E[Y_q] = n_{Q,k}/N_Q$, where we exploited the fact that $n_{Q,k} = \sum_{h=1}^r\frac{1}{h}\cdot{v_h}$.
Then we can build an unbiased estimator for $n_{Q,k}$ as $\hat{n}_{Q,k} = z\cdot\frac{N_Q}{s}\, .$
Indeed, we have $\E[\hat{n}_{Q,k}] = \frac{N_Q}{s}\cdot s\E[Y_q]=n_{Q,k}$.

We can now derive an error bound using Hoeffding's inequality as follows:
\begin{align*}
    \Pr[|\hat{n}_{Q,k}-n_{Q,k}|\geq t]= \Pr\left[\left|\sum_{q=1}^s\frac{Y_q {N_Q}}{s}-n_{Q,k}\right|\geq t\right]\leq 2\exp\left(-\frac{2t^2}{\sum_{q=1}^s(M_q-m_q)^2}\right)\leq 2\exp\left(-\frac{2t^2s}{{N_Q^2}}\right).
\end{align*}
where $m_q$ and $M_q$ are the lower and upper bounds, respectively, on the variable $Y_q\cdot\frac{N_Q}{s}$. 
The last inequality follows from the fact that $m_q=0$ and $M_q = N_Q/s$.
Hence, a good estimate of $n_{Q,k}$ can be obtained with probability $\geq 1 - \delta$ using any $s\geq \frac{N_Q^2}{2 t^2}\ln(2/\delta)$. 
We can then fix $t= \varepsilon n_Q$, which allows us to pick any $s\geq \frac{N_Q^2}{2 (\varepsilon n_Q)^2}\ln(2/\delta)$. Now, since $N_Q\leq r\cdot n_Q$, we conclude that $\hat{n}_{Q,k}\in [n_{Q,k}-{\varepsilon} n_{Q}, n_{Q,k}+{\varepsilon} n_{Q}]$ holds with probability at least $1-\delta$ for any $s\geq\frac{r^2}{2 \varepsilon^2}\ln(2/\delta)$.
\end{proof}
Overall, with this approach query time amounts to searching exact solutions for  $s\geq{\frac{r^2}{2\varepsilon^2}\ln(2/\delta)}$ triplets: here, querying a triplet $(i, j, *)$ means checking $u_i$'s presence in the $r-1$ regions of $Q$ other than $r_j$, where each region takes binary-search time up to $O(\log n_Q)$ if only $\overline{T}$ is produced at the preprocessing.
It follows that an upper bound on query time is $\BO{\frac{r^3}{\varepsilon^2}\log(1/\delta)\log{n_Q}}$, while the space complexity of the approach is $O(N)$, e.g. the cost of storing the input.
Additionally, note that query time can be further reduced in expectation by storing auxiliary data structures, allowing for \emph{faster-than-binary-search} expected-time user lookups within a single region.
Finally, we point out that this straightforward sampling scheme works for any $r$ and $k$.

\subsection{Approximation via Sketch}\label{sec:approxsketching}
We now propose a succinct data structure for the $k$-CLAV problem, which consists of a sketch representing each region. At query time, the sketches of all $r$ regions in the query set $Q$ are merged to provide the estimate of $n_{Q,k}$. 
Let $n_Q$ denote the number of distinct users across all regions of $Q$. 
W.l.o.g., we assume that $k\geq r^2$: it indeed suffices to scale times in $T$ and $k$ accordingly.

The sketch computed for each region is a modified version of the Flajolet-Martin (FM) sketch \cite{flajolet1985probabilistic,alon1996space}, where each bit position of the FM is replaced by a Count-Min sketch consisting of a vector of $\BO{1/\varepsilon}$ counters, each with $\BT{\log r}$ bits, for a suitable parameter $\varepsilon > 0$.
The sketches of regions in $Q$ can be easily merged, by pairwise summing of vector entries.
The underlying FM sketch allows us to estimate the number of distinct users in region $Q$, while the Count-Min sketch allows us to separate the overall number of distinct users from the number of distinct users with a long aggregated visit.

We now describe the sketch $S_j$ for a given region $r_j$.
Let $h: [n]\rightarrow [n]$ and $g_i: [n]\rightarrow [1/\varepsilon]$, for any $i \in [\log n+1]$,  be hash functions randomly chosen from the 2-wise independent hash families $\mathcal H$ and $\mathcal G$, respectively.
For any $\ell\in [\log n+1]$ and region $r_j$, let $V_{\ell,j}$ be a vector of $1/\varepsilon$ counters, each consisting of $\BT{\log r}$ bits and initially set to 0. 
We denote by $\text{tail}(x)$ the number of trailing zeros in the binary representation of the natural number $x$; in other words,  $2^\text{tail(x)}$ is the largest power of two that divides $x$.
The sketch of each region is computed 
as in Algorithm \ref{alg:sketch}. 
Each user $u_i$ in region $r_j$ increases counter $g_{\ell_i}(i)$ in vector $V_{\ell_i,j}$ by $\lfloor \tau_{i,j} r^2 /k \rfloor$, where $\ell_i = \text{tail}(h(i))$, that is, time is now measured in multiples of $k/r^2$ (each user appearing in a region with time lower than $k/r^2$ is ignored).
Note that the counters do not exceed $r^2$.
We observe that any user whose aggregate time is at least $k$ induces at least $r^2-r$ increments to a counter (the $-r$ term is due to the roundings).

\begin{algorithm}[t]
\footnotesize
\caption{Sketch construction for region $r_j$}\label{alg:sketch}
\begin{algorithmic}[1]
\For {Each user $u_i$ in $R_j$}
\If {$\tau_{i,j}\geq k/r^2$}
\State $c_{i,j} = \lfloor \tau_{i,j} r^2 / k \rfloor$
\State $\ell_i = \text{tail}(h(i))$ 
\State $V_{\ell_i,j}[g_{\ell_i}(i)] = \min \left\{r^2, V_{\ell_i,j}[g_{\ell_i}(i)]+ c_{i,j}\right\}$
\EndIf
\EndFor
\end{algorithmic}
\end{algorithm}

For a given query set $Q$, we perform the operations in Algorithm \ref{alg:sketchquery}. 
We first construct the sketch $S_Q$ by combining the sketches $S_j$ for all $r_j \in Q$.
We know that for every $\ell \in [\log n+1]$, each $S_j$ contains a vector $V_{\ell, j}$ of $\BO{1/\varepsilon}$ counters, each of length $\BT{\log{r}}$ bits.
The corresponding vectors in $S_Q$ are $\tilde V_{\ell}$ for $\ell \in [\log n+1]$. We start by setting $\tilde V_{\ell}[i] = \min\{r^2, \sum_{j\in Q} V_{\ell,j}[i]\}$ for each index $i \in [1/\varepsilon]$.
Then, we search for the largest integer $\tilde{\ell}$ such that the vector $\tilde V_{\tilde{\ell}}$ contains a counter of value at least $r^2-r$; this value should be what we get from the standard FM when there are no points with aggregate time lower than $k$.
Our estimate is then $\hat{n}_Q=2^{\tilde{\ell}}$. 
We boost success probabilities using the standard median trick, that is, by keeping $\BO{\log (1/\delta)}$ independent repetitions of the sketch and taking the median of the query results.

\begin{algorithm}[t]
\footnotesize
\caption{Query procedure with query $Q$}\label{alg:sketchquery}
\begin{algorithmic}[1]
\State Let $\tilde V_{\ell}$ for $\ell\in [\log n+1]$ be a vector of $\BO{1/\varepsilon}$ counters of $\BT{\log r}$ bits \Comment{sketch $S_Q$}
\For{each $\ell\in [\log n+1]$ and every $i\in[1/\varepsilon]$}
    \State Set $\tilde V_{\ell}[i] = \min\{r^2, \sum_{j\in Q} V_{\ell,j}[i]\}$ \Comment{$S_Q =$ sum of the sketches $S_{j}$ for all $r_j\in Q$}
\EndFor
\State Let $\tilde{\ell}\in [\log n+1]$ be the largest integer such that vector $\tilde V_{\tilde{\ell}}$ has at least one counter not smaller than $r^2-r$
\State \Return $2^{\tilde{\ell}}$~ $(= \hat{n}_Q)$
\end{algorithmic}
\end{algorithm}

\begin{theorem}\label{thm:sketch}
Consider a query set $Q$ containing $r$ regions, and let $n_{Q,k}$ and $n_{Q,k}^-$ be the number of users that appear in the $r$ regions with an aggregate time of at least $k$ and $k(1-1/r)$, respectively. 
Let $\hat{n}_Q$ be the value returned by Algorithm \ref{alg:sketchquery}.
Then, with a probability of at least $1-\delta$, we have $n_{Q,k}/3 \leq  \hat{n}_Q \leq 3 n_{Q,k}^-$ if $\varepsilon = \BO{n_Q / n_{Q,k}}$, and $n_{Q,k}/3 \leq  \hat{n}_Q \leq 3 n_{Q,k}^- +\BO{\varepsilon n_Q}$ otherwise.
The data structure requires $\BO{m \varepsilon^{-1} \log n \log r \log (1/\delta)}$ bits and $\BO{\varepsilon^{-1} \log n \log (1/\delta)}$ query time.
\end{theorem}
\begin{proof}
We say that user $u_i$ is $(r^2-r)$-long 
if $\sum_{r_j\in Q}\lfloor \tau_{i,j}r^2/k\rfloor \geq r^2-r$.  
Due to floor operations, all users with aggregate time at least $k$ is a $(r^2-r)$-long user.
We observe that for each $(r^2-r)$-long user $u_i$ we have $\tilde V_{\ell_i}[g_{\ell_i}(i)]\geq r^2-r$ where $\ell_i=\text{tail}(h(i))$. 
Let $\hat{n}$ be the number of $(r^2-r)$-long users. Then we have $n_{Q,k} \leq \hat{n}\leq n^-_{Q}$.
Suppose that we remove from the sketches $S_j$ and from $S_Q$ all updates induced by users who are not $(r^2-r)$-long.
All nonzero counters in $S_Q$ will then have a value of at least $r^2-r$, and the index $\tilde{\ell}$ is equivalent to the largest index with a non-zero vector: this is the value extracted by a standard FM sketch with the same hash function $h$.
Therefore, we have that $\hat{n}/3 \leq  2^{\tilde{\ell}+1/2} \leq 3\hat{n}$ with constant probability $1/3$ \cite{alon1996space}, and hence  $n_{Q,k}/3 \leq  2^{\tilde{\ell}} \leq 3n_{Q,k}^-$.

We now consider the contributions of users that are not $(r^2-r)$-long.
We observe that they can affect the value of $\tilde{\ell}$ only if there exists a vector $\tilde V_{\ell'}$ with $\ell'>\tilde{\ell}$ and an index $i'$ such that $\tilde V_{\ell'}[i']\geq r^2-r$.
We expect that $n_Q/2^{\ell'}$ users have a tail of at least $\ell'$; then,
 by the Markov inequality, there are at most $\BO{n_Q/2^{\ell'}}$ users updating $\tilde V_{\ell'}$ with constant probability.
Moreover, by a balls and bins argument, we find that the maximum load of each counter in $\tilde V_\ell$ is $\BO{n_Q \varepsilon / 2^{\ell}}$ with high probability for a sufficiently large $n_Q$.
Since a user who is not $(r^2-r)$-long cannot increase any counter by more than $r^2-r$, no counter in any vector $\tilde V_{\ell}$ with $2^\ell> \BO{n_Q \varepsilon}$ can reach the value $r^2-r$.
Therefore, with constant probability, we will have an additive error of $\BO{n_Q \varepsilon}$.

Each sketch consists of $\BO{\log n}$ vectors, each containing $\BO{1/\varepsilon}$ counters of $\BO{\log r}$ bits, giving $\BO{m \varepsilon^{-1} \log n \log r}$ total bits.
Since two counters can be summed in constant time, the query time is $\BO{r \varepsilon^{-1} \log n}$.
By including the cost of the median trick, we get the stated bounds.
\end{proof}

We conjecture that, by using sketches for counting distinct items that provide an $\varepsilon$-approximation, as in \cite{baryossef2002countdistinct}, we can construct a sketch for $(k,r)$-CLAV that yields an $\epsilon n_Q$ approximation for any $\varepsilon>0$; this would make the sketch size proportional to $\varepsilon^{-2}$.

\section{Geometric-CLAV}\label{sec:geometric}

In this section, we study the Geometric-CLAV problem (Definition~\ref{def:geometric}). 
We will assume input points to be in $\mathbb{R}^d$  with $d=O(1)$ for ease of exposition. 
In Section \ref{sec:geolbuncon}, we provide the unconditional space lower bound.
In Sections \ref{sec:condlb1} and \ref{sec:condlb2}, we prove two conditional lower bounds, that depend on the fast-query version of the $r$-SetDisjointness Conjecture and on the Boolean-Matrix-Multiplication (BMM) conjecture\footnote{The BMM conjecture is a 
        different prominent conjecture and it is independent of the multi-SetDisjointness conjecture. It thus provides additional evidence for the difficulty of Geometric CLAV problem.}~\cite{subcubic,bansal2009regularity}, respectively.
Finally, in Section \ref{sec:ubgeometric}, we provide two data structures based on tabulating and on a reduction to the colored dominance counting problem.

\subsection{An Unconditional Space Lower Bound}\label{sec:geolbuncon}

Here, we prove an unconditional lower bound for the space complexity of the Geometric-CLAV problem.
We first define our point set (the locations of the users). By $X_{i:s}$ we denote a point which has all of its coordinates set to zero except the $i$-th coordinate which is set to value  $s$.
Let $m'=\frac{m}{2d}$.
Define a point set $X_i = \{X_{i:s}|  s \in \mathbb{Z}, -m' \le s \le m',s\not = 0\}$.
For example, for $d=2$, $X_1 = \{(-m',0), (-m'+1,0), \ldots, (-1,0), (1,0), \ldots, (m',0)\}$. 
Each $X_i$ contains $2m'$ points and we define the set $X=X_1 \cup \ldots \cup X_d$. $X$ has 
$m$ points. 

Consider a sequence of $2d$ non-negative integer values $\sigma = (\ell_1, r_1, \ell_2, r_2, \ldots, \ell_d, r_d)$. 
We say $\sigma$ is valid if $1 \le \ell_i, r_i \le m'$ for each $1 \le i \le d$. 
We can now define the rectangle $Q(\sigma) = [-\ell_1,r_1] \times \ldots \times [-\ell_d, r_d]$. 
Let $\Q$ be the set of all the rectangles obtained in this way over all valid sequences.
Observe that $|\Q| = m'^{2d}$ and that
each rectangle in $\Q$ has exactly $2d$ points on its boundary.

\begin{theorem}\label{thm:geolbuncon}
    Any data structure for the Geometric-CLAV problem requires $\min\{n,(m/2d)^{2d}\}$ bits of storage
    regardless of the query time.
\end{theorem}
\begin{proof}

    We use the set $\Q$ defined above.
    Let $b=  \min\{n,|\Q|\}$.
    The main idea is to show that given a sequence of $b$ bits, we can use a data structure to encode the bits such that 
    every bit can be recovered (decoded) by issuing queries.
    This clearly shows that the data structure must consume at least $b$ bits of storage which proves the theorem. 
    Let $B_1, \ldots, B_b$ be an arbitrary sequence of bits. 
    
    We first order the rectangles in $\Q$ by the number of points contained in their interior, breaking ties lexicographically.
    In particular, consider two rectangles $Q(\sigma_1), Q(\sigma_2) \in \Q$ corresponding to two valid sequences $\sigma_1$ and $\sigma_2$. Then $Q(\sigma_1)$ appears before $Q(\sigma_2)$ if $Q(\sigma_1)$ has strictly fewer points than $Q(\sigma_2)$ in its
    interior, but if $Q(\sigma_1)$ and $Q(\sigma_2)$ have equal number of points in their interior, then $Q(\sigma_1)$ appears before $Q(\sigma_2)$ if $\sigma_1$ is lexicographically smaller than $\sigma_2$, and vice-versa otherwise.
    
    Consider the $i$-th rectangle $Q_i$ in $\Q$ in this ordering.
    If $B_i = 0$ then we do nothing. Otherwise, we define a unique user, $u_Q$, that spends exactly one unit of time in each of the
    regions $p_1, \ldots, p_{2d}$ which lie on the boundary of $Q_i$.
    Since $b = \min\{n,|\Q|\}$ this is well-defined. 
    This set of users and the points in $X$ are given as input to the data structure.

    It remains to show that we can decode all the bits using the data structure. 
    Now consider the query algorithm and assume that we would like to extract the value of bit $B_i$.
    Observe the sets $X$ and $\Q$ are fixed (do not depend on the sequence of bits).
    Thus, at the time of decoding both $X$ and $\Q$ can be computed, including the ordering of $\Q$ that was used during the encoding. 
    To extract $B_i$,  we consider the $i$-th rectangle $Q_i$ in $\Q$.
    We query the data structure with $Q_i$ and obtain the number $z$ of users who have spent at least $2d$ time in $Q_i$.
    Let us first consider the case when $Q_i$ contains exactly $2d$ points (including the boundary). 
    If $z=0$, then we can conclude $B_i=0$ since no user was added to spend 1 unit of time at the $2d$ points on the boundary of $Q_i$.
    Otherwise, if $z=1$, then we must have $B_i=1$ because we cannot have $2d$ instances of another user $u_{Q}$ be contained in $Q_i$. 
    Now consider the case when $Q_i$ contains additional points in its interior.
    In this case, we use the fact that we have extracted all the bits $B_1, \ldots, B_{i-1}$. 
    Consider a rectangle $Q' \in \Q$ such that $Q'$ is inside $Q_i$. 
    As we have the value of the bits $B_1, \ldots, B_{i-1}$, we know whether a user $u_{Q'}$ was added or not. 
    For every such rectangle $Q'$, we count the number of users that were added using the previously decoded bits.
    By subtracting this value from $z$, we can decide whether the user $u_{Q_i}$ was added or not which in turn allows us to decode $B_i$.
\end{proof}

\subsection{A Conditional Lower Bound from Fast-Query SetDisjointness}\label{sec:condlb1}

In \refthm{geolbuncon}, to create the worst-case lower bound of $\Omega(m^{2d})$ (we assume $d=O(1)$ hereon), we need to create an instance where
there are $\Omega(m^{2d})$ users. 
While the lower bound rules out getting efficient solutions for large values of $n$, it leaves the question of obtaining
efficient solutions for small values of $n$ open. 
Here, we show that in dimensions 2 and above, very efficient solutions are unlikely. Our source of hardness is the following conjecture. We will change the notation from $r$-SetDisjointness in Section 3 to $\beta$-SetDisjointness below, since $r$ is also used for regions.

\begin{conjecture}[\textit{Fast-Query-SetDisjointness Conjecture}:] Let $S_1,\cdots,S_a$ be a collection of $a$ sets from the universe
$\mathcal{U}$.
Let $\beta$ and $s$ be fixed integer constants. Then, there exists a constant $c$ (depending only on $s$ and $\beta$) such that the following holds: answering $\beta$-SetDisjointness queries when 
$|\mathcal{U}| \geq$ $\log^c a$, in time $O(\log^s a)$ requires $\Omega(a^{\beta-o(1)})$ space in the worst case.
\end{conjecture}

The above conjecture is not explicitly mentioned in its exact form above.
However, in~\cite{goldstein2019hardness} it is proven that for $\beta=2$ the above conjecture follows from the Strong SetDisjointness conjecture. 
In particular, it is proven that one either needs $\Omega(a^{2-o(1)})$ space or a query time of $\Omega(|\mathcal{U}|^{0.5-o(1)})$
and since in our formulation, we can simply choose $c>2s$,  we get a space lower bound~\cite{goldstein2019hardness}.

\vspace{3mm} We show that our construction in \refsec{geolbuncon} can be adapted to obtain a lower bound  assuming Conjecture 3. 
Consider an input to the SetDisjointness Problem, following the notation used in the above conjecture: 
Let $S_1, \ldots, S_a$ be $a$ sets from a universe  $\mathcal{U}$.
Let $\overline{S}_i = \mathcal{U} \setminus S_i$ be the complement of $S_i$.
We will use the same point set $X$ in \refsec{geolbuncon} but with some small modifications. 
First, we set $m=4da$ and thus $m' = 2a$.
We define the set $X_i^+ = \{ X_{i:s} \mid 1 \le s \le 2a\}$.
For $i\not = d$, we define $X_i^- = \{ X_{i:s} \mid -2a \le s \le -1\}$.
Finally, we define $X_d^- = \{X_{i:s} \mid -2da  \le s \le -1\}$.
We call each of subsets $X_i^+$ or $X_i^-$ a \textit{chunk}. Thus, we have $2d$ chunks and
we designate $X_d^-$ as the last chunk and the other chunks are ordered arbitrarily from 1 to $2d-1$.
Observe that the last chunk contains $2da$ points whereas all the other chunks contain $2a$ points. 
Let $X$ be the union of all these chunks and so $|X| = O(a)$.
Given a set $S \subset \mathcal{U}$, and a point $q \in X$, \textit{encoding $S$ at $q$} corresponds to the following operation:
for every index $i\in S$, we create a triplet where the user $u_i$ spends unit time at $q$. 
With the above definition, we build the following instance of $d$-dimensional Geometric-CLAV problem: for every $q \in X_d^-$, we encode $\mathcal{U}$ at $q$.
In other words, $\mathcal{U}$ is encoded at every point of the last chunk.

Next, for every  $q \in X \setminus X_d^-$, we do the following. Recall that by construction, only one coordinate of $q$ is non-zero. 
Let $i$ be the absolute value of this non-zero coordinate of $q$. If $i$ is odd, then we encode $S_{(i+1)/2}$ at $q$ (e.g., if $i=1$, we encode $S_1$).
If $i$ is even, then we encode $\overline{S}_{i/2}$ at $q$. 
Intuitively, this encoding can be described as follows:
For every chunk other than the last chunk, as we go over the points of the chunk away from the origin, 
$i$ increases from 1, 2, 3, 4, $\ldots$ to $2a$, and thus we encode
the sets $S_1, \overline{S}_1, S_2, \overline{S}_2, \ldots, \overline{S}_a$ in that order.
We set $k=2da+1$ and then store this instance in a data structure for Geometric CLAV queries.

Now consider a query to the SetDisjointness problem, given  by
the indices, $i_1, \ldots, i_{2d-1}$, of the sets we are to intersect. 
Observe that each set $S_i$ has been encoded at $2d-1$ points, at one point of every chunk except the last chunk.
We mark some of these points as follows.
For $S_{i_j}$, we consider the $j$-th chunk and its unique point $q_{i_j}$ such that $S_{i_j}$ is encoded at $q_{i_j}$.
Let $c_j$ be the magnitude of the non-zero coordinate of the point $q_{i_j}$. 
By our construction,  $c_j$ is odd and $i_j = \frac{c_j+1}{2}$.
Define $w = \sum_{j=1}^{2d-1}(i_j-1)$.
Observe that $0 \le w \le (2d-1)(a-1)$.
Finally, we mark the point in the last chunk whose non-zero coordinate has value $-(k-w-2d+1)$.
Observe that $1 \le k-w-2d+1 \le 2ad$ and thus such a point exists and thus it can be marked.
We create a query rectangle $Q$ which has all these marked points on its boundary.
Observe that for every $j$, $1 \le j \le 2d-1$, when we look at the $j$-th chunk, we have the encoding of the sets
$S_1, \overline{S}_1, \ldots, S_{i_j-1},\overline{S}_{i_j-1}$ inside $Q$ and thus every user spends exactly
$i_j-1$ time on these points. 
Thus, when we look at the points that are fully inside $Q$ and they lie on the chunks 1 to $2d-1$,
every user spends exactly a total of $w$ time.
In the last chunk, we have exactly $k-w-2d+1$ points in $Q$ (including the boundary) and thus every user
spends exactly $k-w-2d+1$ time.
Call the points that lie on the boundary of $Q$ and in the chunks $1$ to  $2d-1$ the \textit{important points}.
Thus, ignoring the important points, every user spends exactly $k-2d+1$ time in $Q$. 
Finally, observe that if $S_{i_1} \cap S_{i_2} \cap \ldots S_{i_{2d-1}}$ is non-empty, then at least one user spends
one time on each of the $2d-1$ important points, for an aggregate time of $k$ spent in $Q$. 
On the other hand if $S_{i_1} \cap S_{i_2} \cap \ldots S_{i_{2d-1}}$ is empty, then by construction no such user will exist.
This means that if the geometric CLAV query returns a non-zero answer, then we can conclude that the intersection of the sets is non-empty, and empty otherwise. 
Now using the Fast-Query-SetDisjointness Conjecture (Conjecture 3), we get the following result.
\begin{theorem}\label{thm:condlb1}
    Assume that Conjecture 3 holds. 
    Then answering $d$-dimensional geometric CLAV queries in polylogarithmic query time, requires
    $\Omega( m^{2d-1-o(1)})$ space, assuming $d$ is fixed.
\end{theorem}

\subsection{A Conditional Lower Bound for $d=2$ using the Boolean-Matrix-Multiplication Conjecture}\label{sec:condlb2}

Here, we show a different conditional lower bound for the planar version.
This result is slightly weaker than the lower bound presented in the previous section but it uses a different 
conjecture which gives additional evidence that obtaining a very efficient solution for the geometric
CLAV problem is unlikely for another reason. 
Our lower bound is inspired by the lower bound of \cite{kaplan2008efficient} but it uses a conjecture related to
boolean matrix multiplication (BMM)~\cite{subcubic}, defined below.

\begin{definition}
   Let $A$ and $B$ be two $m\times m$ matrices with entries in $\{0,1\}$. The BMM problem is to  compute another $m \times m$ matrix $C$
   such that 
   $
   C[i,j] = \texttt{OR}_{k=1}^m (A[i,k]\, \texttt{AND}\, B[k,j]).
   $
   It is essentially similar to normal matrix multiplication where $+$ and $\times$ operations are replaced by logical \texttt{OR}
   and \texttt{AND}.
\end{definition}

The BMM problem can be solved in $O(m^3)$ time with the standard matrix multiplication algorithm.
The main conjecture that we will use is that this solution cannot be improved significantly using any ``combinatorial'' algorithm.
The exact definition of a ``combinatorial'' algorithm is not universally agreed upon~\cite{bansal2009regularity}, but generally speaking,
it applies to algorithms that only use combinatorial structure and not the algebraic, Strassen-like methods for fast 
matrix-multiplication~\footnote{To quote Bansal and Williams~\cite{bansal2009regularity} ``[T]he term
is mainly just a way of distinguishing those approaches which are different from the algebraic approach
originating with Strassen [\ldots] we simply think of a combinatorial
algorithm as one that does not call an oracle for ring matrix multiplication.''}. 

\begin{conjecture}
    The BMM problem cannot be solved by a combinatorial algorithm in $O(m^{3-\varepsilon})$ time for any constant $\varepsilon > 0$~\cite{subcubic}.
\end{conjecture}

Now  we show that an efficient data structure for our problem can lead to improved combinatorial algorithms for BMM problem. Let $A$ and $B$ be the input to the BMM problem. 
We create $2m$ points as follows. 
Consider the line segment $\ell_0$ between points $(0,1)$ and $(1,2)$ and the line segment
$\ell_1$ between $(1,0)$ and $(2,1)$.
Place $m$ points on each line segment.
Let $p_1, \ldots p_m$ be the points on $\ell_0$ and $q_1, \ldots, q_m$ be the points on $\ell_1$.
These points form our set $X$.
The main property of this point is that for every point $p_i$ on $\ell_0$ and every point $q_j$ on $\ell_1$, there exists
a rectangle that contains $p_i$ and $q_j$ and no other point of $\ell_0$ or $\ell_1$.
We also create $m$ users $u_1, \ldots, u_m$.

Now consider matrix $A$. If $A[i,k]=1$, then $u_k$ spends unit time at point $p_i$.
If $B[k,j]=1$, then $u_k$ spends unit time at point $q_j$.
This creates an instance where there are at most $2m^2$ triplets.
Now, observe that for any two points $p_i$ and $q_j$, we can draw a rectangle that only contains
these two points. 
Observe that if any user $u_k$ spends 2 units of time in this rectangle, then, we must have $A[i,k] = B[k,j] = 1$ which 
in turn implies that $C[i,j]=1$.  If no such user exists, the $C[i,j]=0$. 
Thus, we can solve BMM problem using $m^2$ queries. 
Thus, we get the following lower bound.

\begin{theorem}\label{thm:condlb2}
    Assume that Conjecture 4 holds. Then any data structure that solves the Geometric-CLAV problem in two dimensions for an input of $m$ users, $m$ points and $O(m^2)$ triplets, must either have $\Omega(m^{1-\varepsilon})$ query time
    or have $\Omega(m^{3-\varepsilon})$ preprocessing time, for any $\varepsilon > 0$.
\end{theorem}

\subsection{Upper Bounds for Geometric-CLAV}\label{sec:ubgeometric}

We first address the tightness of our unconditional lower bound, assuming $n$ is larger than $(m/2d)^{2d}$.

\subsubsection{The Tabulation Solution}\label{sec:geomtabulation}

We observe that the problem can be solved with a look-up table using tabulation. 

\begin{theorem} \label{thm:tabulation}
    The geometric CLAV can be solved with 
    $O(m^{2d})$ space and $O(\log m)$ query time.
\end{theorem}
\begin{proof}
    We use the classical rank-space reduction: for every dimension $i$, we replace the $i$-th coordinate of every point in $R$ with its rank.
    This reduces the problem to the case when the points in $R$ have coordinates between $1$ and $m$. 
    Observe that in this case, there are at most $m^{2d}$ combinatorially  different queries and thus, for every query we can simply store the 
    result in a data structure of size $O(m^{2d})$. Given the query rectangle, we simply find the rank of the coordinates of the rectangle in $O(\log m)$ time
    and simply look up the answer. 
\end{proof}

The simple data structure above shows that our unconditional lower bound in Section~\ref{sec:geolbuncon} is tight for ``large'' $n$. We now aim for the case when $n$ is ``small''. 

\subsubsection{An Efficient Data Structure Using Colored Dominance Counting}\label{sec:geom1dim}

We start with the simplest case when $d=1$, that is, all points (locations) in the Geometric-CLAV input are located on the real line. As before, we do a reduction to rank space by simply sorting the points in $R$. 
Thus, we can assume that the points in $R$ are the points $1$ to $m$. 
For an interval $\ell$, let  $\tau_{u,\ell} = \sum_{x \in \ell, \tau_{u,x}\in T} \tau_{u,x}$ denote the total time user $u$ spends in interval $\ell$.  
A \textit{minimal interval} for $u$ is an interval $[l,r]$ such that 
$\tau_{u,[l,r]} \geq k$, 
$\tau_{u,(l,r]} < k $,  and $\tau_{u,[l,r)} < k$.

Let $M_u$ denote the set of all the minimal intervals for user $u$. The next lemma bounds the number of minimal intervals.

\begin{lemma}
    The total number of minimal intervals is at most $N$.
\end{lemma}
\begin{proof}
    For each user $u$, the minimal intervals in $M_u$  can overlap but the intervals in $M_u$ will have distinct starting points 
    and in addition, at each starting point, the user $u$ must have spent strictly more than 0 time. 
    Thus, the total number of minimal intervals is upper bounded by $N$.
\end{proof}

We will now ``lift'' the 1D input to two dimensions. Given two points $p=(p_x,p_y)$ and $q=(q_x,q_y)$ in $\mathbb{R}^{2}$, we say that $p$ dominates $q$ if $p_x \geq q_x$ and $p_y \geq q_y$. If $p \neq q$ one of these inequalities must be strict. We now map an interval $\ell = [a,b]$ to the point $(-a,b)$ in the plane. This ensures that
if an interval $\ell'=[a',b']$ contains $\ell$, then the mapped point $(-a',b')$ dominates the point $(-a,b)$.
Consider now the following problem, called the \textit{colored dominance counting} in 2D.

\begin{definition}[2D Colored Dominance Counting]
   Consider a set $S$ of $N$ points such that each point $p$ is assigned a color $c(p)$ from $[n]$.
   The goal is to store $S$ in a data structure such that given a query point $q=(q_x,q_y)$, one can find the number
   of colors dominated by $q$ or to be more precise, we need to return
   \[
       c(q) = | \{ c(p) \mid p=(a,b) \in S,  a \le q_x , b \le q_y \}.
   \]
\end{definition}

It is known that this problem can be solved with $O(N)$ space and $O(\log_w (n_{Q,k}))$ query time~\cite{gao2023adaptive} where
$w$ is the word size and $n_{Q,k}$ is the answer returned.
Thus, we get the following.
\begin{theorem}\label{thm:1d}
    The geometric-CLAV problem can be solved in 1D with 
    $O\left( \min\{N,m^2\} \right)$ space and with query time $O(\log_w (n_{Q,k}))$  where $w$ is the word size 
    and $n_{Q,k}$ is the answer to the query. 
\end{theorem}

\subsubsection{Higher Dimensions}\label{sec:geomhigherdim}

\begin{theorem}\label{thm:geometric}
    The geometric CLAV can be solved in $d$ dimensions with 
    $O\left( \min\{Nm^{2d-2},m^{2d}\} \right)$ space and with  query time $O(\log_w (n_{Q,k}))$ where $w$ is the word size 
    and $n_{Q,k}$ is the answer to the query. 
\end{theorem}
\begin{proof}
    Let $Q$ be a query rectangle. Note that w.l.o.g, we can assume every boundary of $Q$ passes through a point in $X$
    (otherwise, we can shrink $Q$ without changing the number of points in $Q$). 
    We  tabulate over dimensions 2 to $d$: for each dimension between 2 and $d$, we consider every possible 
    coordinate of the two boundaries of $Q$. 
    Note that for every dimension we have at most $m^2$ for a total of $m^{2d-2}$ choices. 
    For every choice, we build a different data structure. 
    Thus, given $Q$, we  look at the coordinates of $Q$ along dimensions $2$ to $d$ and find the data structure built for $Q$.
    This means that we can simply project $Q$ into the first dimension and use the 1D data structure of \refthm{1d}.
\end{proof}

The answer $n_{Q,k}$ to the query takes at least $(\log n_{Q,k})/w$ words to be stored, and so the runtime is a factor $w/\log w$ from the optimal time. 
Note that if a user and location ID have to fit in a word, we have that $w$ is at least $\log (n+m)$. $n_{Q,k}$ is at most $n$, so in the worst case the run time is $O(\log n /\log \log (n+m))$.
Our conditional lower bounds show that the blow up in our data structure is almost optimal and it cannot be fully avoided.
Our lower bounds build an instance with $N=O(m)$ with still polylogarithmic number of users and it implies a
conditional lower bound of $\Omega(m^{2d-1-o(1)})$ space.
Consequently, the space consumption of the above theorem cannot be improved to $O(N m^{2d-2-\varepsilon})$ for
any constant $\varepsilon > 0$.

Finally, although the above exact data structure is constructed for a fixed $k$, we can build copies for different values of $k$ to get an approximate data structure for all $k$. That is, we store the data structure for each $k' = 2^i$ for $i=1, 2, \ldots, \log K$, where $K$ is the maximum allowed value. This returns the number of users who spent at least $k/2$ time in the query region.

\makeatletter
\if@ACM@anonymous
  
\else
    \subsection*{Acknowledgment} 
This material is based upon work performed while attending AlgoPARC Workshop on Parallel Algorithms and Data Structures at the University of Hawaii at Manoa, in part supported by the National Science Foundation under Grant No. 2452276.
This work was supported in part by: MUR
PRIN 20174LF3T8 AHeAD project; MUR PNRR CN00000013
National Center for HPC, Big Data and Quantum Computing; Marsden Fund (MFP-UOA2226); Danish Research Council (DFF-8021-002498).
F. Silvestri would like to thank B. Zamengo of Motion Analytica for useful discussions on mobility data that inspired this work. 
\fi
\makeatother

\bibliographystyle{alpha}
\bibliography{biblio}

\end{document}